\documentclass[preliminary,copyright]{eptcs}
\providecommand{\event}{LFMTP 2011} % Name of the event you are submitting to

\usepackage{breakurl}             % Not needed if you use pdflatex only.
\usepackage{amsmath,amssymb,amsthm}
\usepackage{stmaryrd,mathpartir}

\theoremstyle{plain}% actually the default style
\newtheorem{theorem}{Theorem}[section]
\newtheorem{lemma}[theorem]{Lemma}
\newtheorem{corollary}[theorem]{Corollary}

\theoremstyle{definition}
\newtheorem{definition}[theorem]{Definition}
\newtheorem{example}[theorem]{Example}

\newcommand{\act}{\cdot}
\newcommand{\acts}{\ast}
\newcommand{\arrow}{\rightarrow}
\newcommand{\as}{\overline{a}}
\newcommand{\Atom}{\mathbb{A}}
\newcommand{\bimp}{\Leftrightarrow}
\newcommand{\comment}[1]{}
\newcommand{\defeq}{\;\triangleq\;}

\newcommand{\dom}{\mathit{dom}}
\newcommand{\ds}{\mathit{ds}}
\newcommand{\ent}{\vdash}

\newcommand{\ento}{\vdash^o}
\newcommand{\eq}{\approx}
\newcommand{\fe}[1][\nabla]{#1}
\newcommand{\FE}{\mathsf{FE}}
\newcommand{\freshfor}{% code provided by Robin Fairbairn
  \ensuremath{\mathrel{%
      \hbox{\ooalign{%
          \hfil$\approx$\hfil\cr%
          \hfil$\sslash$\hfil}}}}}
\newcommand{\idp}{\iota}
\newcommand{\imp}{\Rightarrow}
\newcommand{\nsupp}{\mathrel{\#}}
\newcommand{\op}[1][op]{\mathit{#1}}
\newcommand{\Op}{\mathsf{Op}}
\newcommand{\Perm}{\mathsf{Perm}}
\newcommand{\pfun}{\rightharpoonup}
\newcommand{\Pow}{\mathcal{P}}
\newcommand{\ruleref}[1]{(\RefTirName{#1})}
\newcommand{\se}[1][\Gamma]{#1}
\newcommand{\SE}{\mathsf{SE}}
\newcommand{\seof}[1]{#1^{:}}
\newcommand{\sff}{\mbox{\scriptsize$\freshfor$}}
\newcommand{\Sig}{\Sigma}
\newcommand{\sort}[1][s]{\mathsf{#1}}
\newcommand{\Sort}{\mathsf{Sort}}
\newcommand{\sub}{\sigma}
\newcommand{\supp}{\mathit{supp}}
\newcommand{\swap}[2]{(#1\;#2)}
\newcommand{\Term}[3]{#1_{#2}(#3)}
\newcommand{\Th}{\mathbb{T}}
\newcommand{\tm}{\mathsf{tm}}
\newcommand{\Var}{\mathsf{Var}}

\title{Nominal Logic with Equations Only}
\author{Ranald Clouston
\thanks{The author gratefully acknowledges the advice of Andrew Pitts and the
reviewers, and the support of the Woolf Fisher Trust.}
\institute{Logic and Computation Group, Research School of Computer Science,
Canberra, ACT 0200, Australia}
\email{ranald.clouston@anu.edu.au}
}
\def\titlerunning{Nominal Logic with Equations Only}
\def\authorrunning{R. Clouston}
\begin{document}
\maketitle

\begin{abstract}
Many formal systems, particularly in computer science, may be captured by
equations modulated by side conditions asserting the ``freshness of names'';
these can be reasoned about with Nominal Equational Logic (NEL). Like most
logics of this sort NEL employs this notion of freshness as a first class
logical connective. However, this can become inconvenient when attempting to
translate results from standard equational logic to the nominal setting. This
paper presents proof rules for a logic whose only connectives are equations,
which we call Nominal Equation-only Logic (NEoL). We prove that NEoL is just as
expressive as NEL. We then give a simple description of equality in the empty
NEoL-theory, then extend that result to describe freshness in the empty
NEL-theory.
\end{abstract}

%%%%%%%%%%%%%%%%%%%%%%%%%%%%%%%%%%%%%%%%%%%%%%%%%%%%%%%%%%%%%%%%%%%%%%%%%%%%%%%
\section{Introduction}

Many formal systems, particularly in computer science, may be captured via
equations modulated by side conditions asserting certain names are \emph{fresh
for} (not in the free names of) certain metavariables:
\begin{description}
  \item[First-order logic:]
  $\Phi\supset(\forall a.\,\Psi)\;=\;\forall a.\,(\Phi\supset\Psi)$ if $a$ is
    fresh for $\Phi$;
  \item[$\lambda$-calculus:]
  $\lambda\,a.\,f\,a \;=_{\eta}\; f$ if $a$ is fresh for $f$;
  \item[$\pi$-calculus:]
  $(\nu a\,\,x)\mid y \;=\; \nu a\,(x\mid y)$ if $a$ is fresh for $y$.
\end{description}
We may express such modulated equations, and hence reason formally about the
systems described by them, with \emph{Nominal Equational Logic
(NEL)} \cite{Clouston:Nominal}. NEL-theories can also express the notions of
binding and $\alpha$-equivalence such systems exhibit \cite{Clouston:Binding}.
NEL generalises standard equational logic by employing the \emph{nominal
sets} model \cite{Pitts:Nominal}, a refinement of the earlier Fraenkel-Mostowski
sets model \cite{Gabbay:New}, where the manipulation of names is modelled by the
action of \emph{permutations}.

In the examples above the `fresh for' relation, represented in NEL by the symbol
`$\freshfor$', is attached to metavariables as a side condition to the
equations. However this relation generalises naturally and conveniently to a
relation asserting certain names are fresh for certain \emph{terms}. As such, in
NEL and other nominal logics, $\freshfor$ is treated as a first class logical 
connective, rather than merely being used in side conditions.

Standard equational logic is an extremely well studied system (e.g.
\cite[Cha. 3]{Burris:Logic}), and NEL's development philosophy was to maintain
as close a relationship as possible to this standard account, so that well known
results can be transferred to a setting with names and binding with a minimum of
difficulty. There are two orthogonal ways in which this can be done. The first
is to translate techniques and results from equational logic to NEL, such as
term rewriting, Lawvere theories, algebras for a monad, and Birkhoff-style
closure results. The second is to combine NEL with other extensions of
equational logic, of which a multitude exist - partial, order-sorted,
conditional, membership, fuzzy, and so forth.

There is a tension here, as we wish to exploit known results of equational logic
while also having the convenience of $\freshfor$ as a first class logical
connective, so that NEL is no longer really an equational logic at all. Consider
the case of Lawvere theories \cite{Lawvere:Functorial}. In this category
theoretic view of equational logic, an equational theory is mapped to a
\emph{classifying category}, whose arrows are tuples of terms, and the equality
of arrows is asserted to correspond to the provable equality of terms in the
theory. This is intuitive because equality has a clear meaning in the category
theoretic setting, as it must in any mathematical setting. This is not true of
freshness, which is a notion bespoke to the nominal sets model with no obvious
meaning in, say, standard category theory.

Fortunately, it is known that in a variety of contexts freshness judgements can
be translated into equally expressive equations with freshness side conditions.
The earliest such result to our knowledge is for the related logic of Nominal
Algebra (see e.g. \cite[Lem. 4.5.1]{Gabbay:Nominal}); the corresponding result
for NEL is Lem. \ref{lem:frsheqeq} of this paper.  We can therefore treat
$\freshfor$ as syntactic sugar, and justifiably call NEL an equational logic.

However, working with the standard proof rules for NEL, in which $\freshfor$ is
extensively used as a first class connective, may still be highly inconvenient
when trying to exploit known results of equational logic. Developing an
analogue of Lawvere theories for NEL in \cite{Clouston:Lawvere} required some
complex proofs relating logical derivations in NEL to category theoretic
properties. Having freshness judgements in those proof derivations with no
obvious category theoretic interpretation, and therefore being forced to apply 
the conversion to equations of Lem. \ref{lem:frsheqeq} each time, would have
been extremely laborious, and render the proofs obscure. Instead, the
development of Nominal Lawvere theories used alternative proof rules for NEL
that employ equations only, relegating freshness assertions back to the side
conditions.

In this paper we present these proof rules (in slightly modified form), which we
call Nominal Equation-only Logic (NEoL), and in Sec. \ref{sec:neol} we show that
NEoL and NEL coincide. This result, which until now has only appeared in Sec.
5.5 of the author's thesis \cite{Clouston:Equational}, is crucial to the
published proof that Nominal Lawvere theories correspond to NEL-theories. We
posit that NEoL will continue to be convenient when applying standard equational
logic results to names and binding, even as NEL remains the more convenient
system for applications.

In standard equational logic two terms are provably equal to each other in the
empty theory if and only if they are syntactically identical. Sec.
\ref{sec:empty} presents for the first time a simple syntax-directed description
of equality in the empty theory for NEoL. Cor. \ref{cor:freshback} extends this
result to NEL to give a description of freshness in the empty theory. Finally
Sec. \ref{sec:related} compares NEoL and NEL to three related notions in the
literature of equational logic over nominal sets.

%%%%%%%%%%%%%%%%%%%%%%%%%%%%%%%%%%%%%%%%%%%%%%%%%%%%%%%%%%%%%%%%%%%%%%%%%%%%%%%
\section{Nominal Sets}

We will first introduce the basic mathematics of the nominal sets model, which
will be necessary for the presentation of the syntax of Nominal Equational Logic
in the next section.

Fix a countably infinite set $\Atom$ of \emph{atoms}, which we will use as
names. The set $\Perm$ of (finite) \emph{permutations} consists of all
bijections $\pi:\Atom\to\Atom$ whose domain
\begin{equation}
  \label{eq:nom1}
  \supp(\pi)\defeq\{a\mid\pi(a)\neq a\}
\end{equation}
is finite. $\Perm$ is a group with multiplication as permutation composition
$\pi'\pi(a)=\pi'(\pi(a))$, and identity as the permutation $\idp$ leaving all
atoms unchanged. $\Perm$ is generated by \emph{transpositions} $\swap{a}{b}$
that map $a$ to $b$, $b$ to $a$ and leave all other atoms unchanged. We will
make particular use of permutations known as \emph{generalised
transpositions}~\cite[Lem. 10.2]{Clouston:Nominal}. Let
\begin{equation*}
%  \label{eq:nom1a}
  \Atom^{(n)} \defeq \{(a_1,\ldots,a_n) \in \Atom^n \mid
    a_i\neq a_j\mbox{ for }1\leq i<j\leq n\}\enspace.
\end{equation*}
All the tuples of atoms we use in this paper will be from this set. Take
$\vec{a}=(a_1,\ldots,a_n),\vec{a}'=(a'_1,\ldots,a'_n)\in\Atom^{(n)}$ with
disjoint underlying sets. Then we define their generalised transposition as
\begin{equation*}
  \swap{\vec{a}}{\vec{a}'} \;\defeq\; \swap{a_1}{a'_1}\cdots\swap{a_n}{a'_n}
    \enspace.
\end{equation*}
A \emph{$\Perm$-set} is a set $X$ equipped with a
function, or $\Perm$-action, $(\pi,x)\mapsto\pi\act x$ from $\Perm\times X$ to
$X$ such that $\idp\act x=x$ and $\pi\act(\pi'\act x)=\pi\pi'\act x$.

Given such a $\Perm$-set $X$ we say that a set of atoms $\as\subseteq\Atom$
\emph{supports} $x\in X$ if for all $\pi\in\Perm$, $\supp(\pi)\cap\as=
\emptyset$ implies that $\pi\act x=x$.

\begin{definition}\label{def:nom}
A \emph{nominal set} is a $\Perm$-set $X$ with the \emph{finite support
property}: for each $x\in X$ there exists some finite $\as\subseteq\Atom$
supporting $x$.

If an element $x$ is finitely supported then there is a unique least such
support set~\cite[Prop.~3.4]{Gabbay:New}, which we write $\supp(x)$ and call
\emph{the support of} $x$. This may be read as the \emph{set of free names} of a
term. If $\as\cap\supp(x)=\emptyset$ for some $\as\subseteq\Atom$ we say that
\emph{$\as$ is fresh for $x$} and write $\as\nsupp x$, capturing the \emph{not
free in} relation.
\end{definition}

\begin{example}\label{ex:nom}
\begin{enumerate}
\item
  Any set becomes a nominal set under the trivial $\Perm$-action $\pi\act x=x$,
  with finite support property $\supp(x)=\emptyset$;
\item
  $\Atom$ is a nominal set with $\Perm$-action $\pi\act a=\pi(a)$ and $\supp(a)=
  \{a\}$;
\item
  $\Perm$ is a nominal set with $\Perm$-action $\pi\act\pi'=\pi\pi'\pi^{-1}$ and
  support as in \eqref{eq:nom1};
\item
  Finite products of nominal sets are themselves nominal sets given the
  element-wise $\Perm$-action and $\supp(x_1,\ldots,x_n)=\bigcup_{1\leq i\leq n}
  \supp(x_i)$;
\item
  $\Atom^{(n)}$, and the set of finite sets of atoms $\Pow_{fin}(\Atom)$, are
  nominal sets given the element-wise $\Perm$-actions. Supports correspond to
  underlying sets.
\end{enumerate}
\end{example}

\begin{lemma}\label{lem:somefacts}
Given a nominal set $X$, element $x\in X$, and permutations $\pi,\pi'\in\Perm$,
\begin{enumerate}
\item
  Given finite $\as\subseteq\Atom$, $\as\nsupp x$ implies $\pi\act\as\nsupp\pi
    \act x$;
\item
  The \emph{disagreement set} of $\pi$ and $\pi'$ is
  \begin{equation*}
    \ds(\pi,\pi') \defeq \{a \mid \pi(a)\neq\pi'(a)\}\enspace.
  \end{equation*}\
  Then $\ds(\pi,\pi')\nsupp x$ implies $\pi\act x=\pi'\act x$.
\end{enumerate}
\end{lemma}
\begin{proof}
\cite[Lem. 3.7]{Pitts:Alpha} and \cite[Lem. 7.3(iv)]{Clouston:Nominal}.
\end{proof}

Given $\Perm$-sets $X,Y$ we can define a $\Perm$-action on functions $f:X\to Y$ 
by
\begin{equation*}
  (\pi\act f)(x) \defeq \pi\act(f(\pi^{-1}\act x))\enspace.
\end{equation*}
Hence if $f$ maps $x\mapsto y$ then $\pi\act f$ maps $\pi\act x\mapsto\pi\act
y$. A \emph{finitely supported function} is a function with finite support under
this definition; this terminology is necessary as even where $X,Y$ are nominal
sets not all functions between them have this property. In the particular case
that $f$ has empty support, we call it \emph{equivariant}. This is equivalent to
the condition that $\pi\act(f(x))=f(\pi\act x)$ for all permutations $\pi$.

%%%%%%%%%%%%%%%%%%%%%%%%%%%%%%%%%%%%%%%%%%%%%%%%%%%%%%%%%%%%%%%%%%%%%%%%%%%%%%%
\section{Nominal Equational Logic}

This section presents syntax and proof rules for Nominal Equational Logic (NEL)
\cite{Clouston:Nominal}. In fact it is sometimes useful to mildly generalise NEL
so that its sorts form a nominal set, rather than a set, of sorts, as is done in
\cite{Clouston:Lawvere}. However this does not materially effect the results of
this paper and so the simpler original presentation is here used.

\begin{definition}\label{def:sig}
A \emph{NEL-signature} $\Sig$ is specified by
\begin{enumerate}
\item
  a set $\Sort_\Sig$, whose elements are called the \emph{sorts of $\Sig$};
\item
  a nominal set $\Op_\Sig$, whose elements are called the \emph{operation
  symbols of $\Sig$};
\item
  an equivariant \emph{typing function} mapping each operation symbol $\op\in
  \Op_{\Sig}$ to a \emph{type} consisting of a finite list $\vec{\sort}=(\sort_1
  \ldots,\sort_n)$ of sorts of $\Sig$ and another $\sort\in\Sort_\Sig$. We write
  this $\op:\vec{\sort}\arrow\sort$. Where $n=0$ we write $\op:\sort$.
  Equivariance of the typing function means that each $\op,\pi\act\op$ have the
  same type.
\end{enumerate}
\end{definition}

\begin{example}\label{ex:lam}
A NEL-signature for the untyped $\lambda$-calculus can be defined by letting
our sorts be the singleton $\{\tm\}$ and operation symbols be
\begin{equation*}
  \{\op[var]_a \mid a\in\Atom\} \cup \{\op[lam]_a \mid a\in\Atom\} \cup
    \{\op[app]\}
\end{equation*}
representing object-level variables, lambda-abstractions and application
respectively. The $\Perm$-action on these operations symbols is
\begin{equation*}
  \pi\act\op[var]_a\defeq\op[var]_{\pi(a)},\quad
    \pi\act\op[lam]_a\defeq\op[lam]_{\pi(a)},\quad
    \pi\act\op[app]\defeq\op[app]\enspace.
\end{equation*}
The typing function is
\begin{equation*}
  \op[var]_a:\tm,\quad \op[lam]_a:(\tm)\to\tm,\quad
    \op[app]:(\tm,\tm)\to\tm\enspace.
\end{equation*}
\end{example}

\begin{definition}
Fix a countably infinite set $\Var$ of \emph{variables}. Then the \emph{terms}
over $\Sig$ are
\begin{equation*}
  t ::= \pi\,x \mid \op\,t\cdots t
\end{equation*}
for $\pi\in\Perm$, $x\in\Var$ and $\op\in\Op_{\Sig}$. We call $\pi\,x$ a
\emph{suspension} and write $\idp\,x$ simply as $x$. We call $\op\,t_1\cdots
t_n$ a \emph{constructed term}.

The \emph{sorting environments} $\SE_{\Sig}$ are partial functions $\se:\Var
\pfun\Sort_{\Sig}$ with finite domain. We define the set
$\Term{\Sig}{\sort}{\se}$ of \emph{terms of sort $\sort$ in $\se$} by
\begin{enumerate}
\item
  if $\pi\in\Perm$ and $x\in\dom(\se)$ then $\pi\,x\in\Term{\Sig}{\se(x)}{\se}$;
\item
  if $\op:(\sort_1,\ldots,\sort_n)\arrow\sort$ and $t_i\in
    \Term{\Sig}{\sort_i}{\se}$ for $1\leq i\leq n$, then $\op\,t_1\cdots t_n\in
    \Term{\Sig}{\sort}{\se}$.
\end{enumerate}
\end{definition}

The \emph{object-level $\Perm$-action} on terms, $(\pi,t\in
\Term{\Sig}{\sort}{\se})\mapsto\pi\acts t\in\Term{\Sig}{\sort}{\se}$, is
\begin{equation}
  \label{eq:nel1}
  \begin{array}{rcl}
    \pi\acts(\pi'\,x) &\defeq& \pi\pi'\,x\enspace; \\
    \pi\acts(\op\,t_1\cdots t_n) &\defeq& (\pi\act\op)(\pi\acts t_1)\cdots
      (\pi\acts t_n)\enspace.
  \end{array}
\end{equation}
This action is used in the definition of substitution: given $\se,\se'\in
\SE_{\Sig}$, a \emph{substitution} $\sub:\se\to\se'$ is a map from each $x\in
\dom(\se)$ to $\sub(x)\in\Term{\Sig}{\se(x)}{\se'}$. Given a term $t\in
\Term{\Sig}{\sort}{\se}$, the term $t\{\sub\}\in\Term{\Sig}{\sort}{\se'}$ is
defined by
\begin{equation}\label{eq:nel11}
  \begin{array}{rcl}
    (\pi\,x)\{\sub\} &\defeq& \pi\acts\sub(x)\enspace; \\
    (\op\,t_1\cdots t_n)\{\sub\} &\defeq&
      \op\,t_1\{\sub\}\cdots t_n\{\sub\}\enspace.
  \end{array}
\end{equation}
We will write the single term substitution that replaces the variable $x$ with
the term $t$ and leaves all other variables unchanged as $(t/x)$.

Terms are not in general finitely supported under the $\Perm$-action
\eqref{eq:nel1}. However there is another notion of $\Perm$-action on terms
which has this property, so that each $\Term{\Sig}{\sort}{\se}$ is a nominal
set. The \emph{meta-level $\Perm$-action} on terms, $(\pi,t\in
\Term{\Sig}{\sort}{\se})\mapsto\pi\act t\in\Term{\Sig}{\sort}{\se}$, is
\begin{equation}
  \label{eq:nel1x}
  \begin{array}{rcl}
    \pi\act(\pi'\,x) &\defeq& \pi\pi'\pi^{-1}\,x\enspace; \\
    \pi\act(\op\,t_1\cdots t_n) &\defeq& (\pi\act\op)(\pi\act t_1)\cdots
      (\pi\act t_n)\enspace.
  \end{array}
\end{equation}
The following Lemma relates these notions:

\begin{lemma}\label{Lem:act_s}
Given $t\in\Term{\Sig}{\sort}{\se}$, $\pi\in\Perm$ and a substitution $\sub$,
\begin{enumerate}
\item
  $\pi\acts(t\{\sub\})=(\pi\acts t)\{\sub\}$;
\item
  $\pi\act t = \pi\acts t\{\pi^{-1}-\}$, where $(\pi^{-1}-)$ is the
  substitution mapping each $x\mapsto\pi^{-1}\,x$.
\end{enumerate}
\end{lemma}
\begin{proof}
Easy inductions on the structure of $t$; see
\cite[Lem 5.2 \& (30)]{Clouston:Nominal} or \cite[Lem. 2.3]{Gabbay:One}.
\end{proof}

The \emph{freshness environments} $\FE_{\Sig}$ are partial functions $\fe$ with
finite domain on $\Var$, mapping each $x\in\dom(\fe)$ to a pair $(\as,\sort)$
where $\as\in\Pow_{fin}(\Atom)$ and $\sort\in\Sort_{\Sig}$. $\FE_{\Sig}$ is then a
nominal set under the action $(\pi\act\fe)(x)=(\pi\act\as,\sort)$; $\supp(\fe)$
is $\bigcup_{x\in\dom(\fe)}\supp(\fe(x))$. If $\fe(x_i)=(\as_i,\sort_i)$ for $1
\leq i\leq n$ we write $\fe$ as
\begin{equation}
  \label{eq:nel13}
  (\as_1\freshfor x_1:\sort_1,\ldots,\as_n\freshfor x_n:\sort_n)\enspace.
\end{equation}
The intended meaning is that $\as_i$ is fresh for $x_i$, which has sort
$\sort_i$. These will capture the freshness side conditions we discussed in the
introduction. Each $\fe\in\FE_{\Sig}$ gives rise to a sorting environment
$\seof{\fe}\in\SE_{\Sig}$ by taking the second projection. We will abbreviate
$\{a\}\freshfor x:\sort$ as $a\freshfor x:\sort$ and $\emptyset\freshfor x:
\sort$ as $x:\sort$.

\begin{definition}\label{def:NELjud}
A \emph{NEL-judgement} has the form
\begin{equation}
  \label{eq:nel18}
  \fe\ent \as\freshfor t\eq t':\sort
\end{equation}
where $\fe\in\FE_{\Sig}$, $\as\in\Pow_{fin}(\Atom)$, $\sort\in\Sort_{\Sig}$ and
$t,t'\in\Term{\Sig}{\sort}{\seof{\fe}}$. We will abbreviate $\fe\ent\as
\freshfor t\eq t:\sort$ as $\fe\ent\as\freshfor t:\sort$ and $\fe\ent\emptyset
\freshfor t\eq t':\sort$ as $\fe\ent t\eq t':\sort$.

A \emph{NEL-theory} $\Th$ is a collection of such judgements, which we call its
\emph{axioms}.
\end{definition}

\begin{example}\label{ex:lam2}
The axioms for  $\alpha\beta\eta$-equivalence over the untyped
$\lambda$-calculus (Ex. \ref{ex:lam}), adapting \cite[Ex. 2.15]{Gabbay:Nominal},
are
\begin{description}
  \item[($\alpha$)]
  $\;(x:\tm)\ent a\freshfor \op[lam]_a\,x:\tm$
  \item[($\beta_1$)]
  $(a\freshfor x:\tm,y:\tm)\ent \op[app]\,(\op[lam]_a\,x)\,y \eq x:\tm$
  \item[($\beta_2$)]
  $(y:\tm)\ent \op[app]\,(\op[lam]_a\,\op[var]_a)\, y \eq y:\tm$
  \item[($\beta_3$)]
  $(x:\tm,b\freshfor y:\tm)\ent \op[app]\,(\op[lam]_a\,(\op[lam]_b\,x))\,y \eq
    \op[lam]_b\,(\op[app]\,(\op[lam]_a\,x)\,y) :\tm$
  \item[($\beta_4$)]
  $(x_1:\tm,x_2:\tm,y:\tm)\ent \op[app]\,(\op[lam]_a\,(\op[app]\,x_1\,x_2))\,y
    \eq\op[app]\,(\op[app]\,(\op[lam]_a\,
    x_1)\,y)\,(\op[app]\,(\op[lam]_a\,x_2)\,y):\tm$
  \item[($\beta_5$)]
  $(b\freshfor x:\tm)\ent \op[app]\,(\op[lam]_a\,x)\,\op[var]_{b}\eq
    \swap{a}{b}\,x:\tm$
  \item[($\eta$)]
  $\;(a\freshfor x:\tm)\ent \op[lam]_a\,(\op[app]\,x\,\op[var]_a)\eq x:\tm
    \enspace.$
\end{description}
\end{example}

\begin{definition}\label{Def:logc}
  [Logical Consequence]
The set of \emph{theorems} of a NEL-theory $\Th$ is the least set of judgements
containing the axioms of $\Th$ and closed under the rules of Fig.
\ref{Fig:rulnel}. We write
\begin{equation*}
  \fe\ent_{\Th} \as\freshfor t \eq t':\sort
\end{equation*}
to indicate that the judgement is a theorem of $\Th$.

Fig. \ref{Fig:rulnel} uses the following new pieces of notation:
\begin{itemize}
\item
  In \ruleref{weak} the relation $\fe\leq\fe'$ holds if $\dom(\fe)\subseteq\dom
  (\fe')$ and for all $x\in\dom(\fe)$ we have $\fe(x)=(\as,\sort)$ and $\fe'(x)=
  (\as',\sort)$ so that $\as\subseteq\as'$.
\item
  In rule \ruleref{subst}
  \begin{equation}
    \label{eq:pro1}
    \fe'\ent\sub\eq\sub':\fe
  \end{equation}
  stands for the hypotheses $\fe'\ent\as_i\freshfor\sub(x_i)\eq\sub'(x_i):
  \sort_i$ for $1\leq i\leq n$, where $\fe$ is as \eqref{eq:nel13}.
\item
  In \ruleref{atm-intro} and \ruleref{atm-elim}, if $\as$ is a finite set of
  atoms and $\fe$ is as \eqref{eq:nel13} then
  \begin{equation*}
    \fe^{\sff\as}\defeq (\as_1\cup\as\freshfor x_1:\sort_1, \ldots,\,
      \as_n\cup\as\freshfor x_n:\sort_n)\enspace.
  \end{equation*}
\end{itemize}
Note also that \ruleref{atm-intro} and \ruleref{atm-elim} carry side conditions
relating to freshness. These do not refer to the freshness connective
$\freshfor$ internal to the logic. Rather, they refer to the
not-in-the-support-of relation $\nsupp$ of Def. \ref{def:nom} over the nominal
sets $\FE_{\Sig}$, $\Pow_{fin}(\Atom)$, and $\Term{\Sig}{\sort}{\seof{\fe}}$ with
respect to the action \eqref{eq:nel1x}.
\end{definition}

\begin{figure}
  \begin{mathpar}
    \inferrule*[left=(refl),right={$\fe\in\FE_{\Sig},
                                    t\in\Term{\Sig}{\sort}{\seof{\fe}}$}]%
    {%
      \ }{%
      \fe\ent t\eq t:\sort }
    \and
    \inferrule*[left=(symm)]%
    {%
      \fe\ent \as\freshfor t\eq t':\sort }{%
      \fe\ent \as\freshfor t'\eq t:\sort } \and
    \inferrule*[left=(trans)]%
    {%
      \fe\ent \as_1\freshfor t\eq t':\sort\\
      \fe\ent \as_2\freshfor t'\eq t'':\sort }{%
      \fe\ent (\as_1\cup\as_2)\freshfor t\eq t'':\sort } \and
    \inferrule*[left=(weak),right={$\fe\leq\fe'\in\FE_{\Sig}$}]%
    {
      \fe\ent\as\freshfor t\eq t':\sort}{%
      \fe'\ent\as\freshfor t\eq t':\sort} \and
    \inferrule*[left=(subst),
    right={$\sub,\sub':\seof{\fe}\to\seof{(\fe')}$}]%
    {%
      \fe'\ent \sub\eq\sub':\fe\\
      \fe\ent \as\freshfor t\eq t':\sort}{%
      \fe'\ent \as\freshfor t\{\sub\}\eq t'\{\sub'\}:\sort} \and
    \inferrule*[left=(atm-intro),right={$\as\nsupp(\as',t,t')$}]%
    {%
      \fe\ent \as'\freshfor t\eq t':\sort}{%
      \fe^{\sff\as}\ent \as'\cup\as\freshfor t\eq t':\sort } \and
    \inferrule*[left=(atm-elim),right={$\as\nsupp(\fe,\as',t,t')$}]%
    {%
      \fe^{\sff\as}\ent \as'\freshfor t\eq t':\sort}{%
      \fe\ent \as'\freshfor t\eq t':\sort } \and
    \inferrule*[left={($\freshfor$-equivar)}]%
    {%
      \ }{%
      (\as\freshfor x:\sort)\ent \pi\act\as\freshfor\pi\,x:\sort} \and 
    \inferrule*[left={(susp)}]%
    {%
      \ }{%
      (\ds(\pi,\pi')\freshfor x:\sort)\ent\pi\,x\eq \pi'\,x:\sort }
  \end{mathpar}
  \caption{Proof rules for NEL}
  \label{Fig:rulnel}
\end{figure}

In \cite{Clouston:Nominal} semantics are given for NEL, in which sorts are
interpreted as nominal sets and operation symbols as finitely supported
functions between them. The proof rules of Fig. \ref{Fig:rulnel} are shown to be
sound and complete for that semantics. In this paper, however, we will work
purely in terms of NEL's proof theory.

The next Lemma shows how freshness judgements may be translated into equivalent
equational judgements in NEL. This will be crucial to the results of the next
section, where we will get rid of freshness judgements entirely.

\begin{lemma}\label{lem:frsheqeq}
Given $t\in\Term{\Sig}{\sort}{\seof{\fe}}$ and $\as\in\Pow_{fin}(\Atom)$,
\begin{equation*}
  \fe\ent_{\Th}\as\freshfor t:\sort \;\bimp\;
    \fe^{\sff\supp(\vec{a}')}\ent_{\Th} t\eq\swap{\vec{a}}{\vec{a}'}\acts t:
    \sort
\end{equation*}
where $\vec{a}\in\Atom^{(n)}$ is an ordering of $\as$ and $\vec{a}'\in
\Atom^{(n)}$ is a tuple of the same size such that $\supp(\vec{a}')\nsupp(\fe,
\as,t)$.
\end{lemma}
\begin{proof}
Left-to-right: $\fe^{\sff\supp(\vec{a}')}\ent\as\cup\supp(\vec{a}')\freshfor t:
\sort$ by \ruleref{atm-intro}; $(\as\cup\supp(\vec{a}')\freshfor x:\sort)\ent x
\eq\swap{\vec{a}}{\vec{a}'}\,x:\sort$ by \ruleref{susp}; the result then follows
by \ruleref{subst} and \eqref{eq:nel11}.

Right-to-left: $\fe^{\sff\supp(\vec{a}')}\ent\supp(\vec{a}')\freshfor t:\sort$ by
\ruleref{refl} and \ruleref{atm-intro}; $(\supp(\vec{a}')\freshfor x:\sort)\ent
\as\freshfor\swap{\vec{a}}{\vec{a}'}\,x:\sort$ by \ruleref{$\freshfor$-equivar};
then $\fe^{\sff\supp(\vec{a}')}\ent\as\freshfor\swap{\vec{a}}{\vec{a}'}\acts t:
\sort$ by \ruleref{subst}. Our hypothesis along with \ruleref{trans} and
\ruleref{symm} gives us $\fe^{\sff\supp(\vec{a}')}\ent\as\freshfor t:\sort$; the
result then follows by \ruleref{atm-elim}.
\end{proof}

\begin{example}\label{ex:lam3}
The rule for $\alpha$-equivalence for the untyped $\lambda$-calculus in Ex.
\ref{ex:lam2}
$$
(x:\tm)\ent a\freshfor \op[lam]_a\,x:\tm
$$
is equivalent to
\begin{equation}
  \label{eq:nel5}
(b\freshfor x:\tm)\ent \op[lam]_a\,x\eq\op[lam]_b\,\swap{a}{b}\,x:\tm\enspace.
\end{equation}
\end{example}

%%%%%%%%%%%%%%%%%%%%%%%%%%%%%%%%%%%%%%%%%%%%%%%%%%%%%%%%%%%%%%%%%%%%%%%%%%%%%%%
\section{Nominal Equation-only Logic}\label{sec:neol}

This section presents syntax and proof rules for NEL without freshness
connectives to the right of the turnstile $\ent$. We call this Nominal
Equation-only Logic (NEoL), and show that it is just as expressive as NEL.

Note that the previously published version of NEoL \cite{Clouston:Lawvere}
included a rule called \ruleref{perm} that was somewhat unwieldy. This paper
improves the presentation of NEoL by replacing \ruleref{perm} with a special
case \ruleref{susp}, and then derives \ruleref{perm} as Lem. \ref{susp->perm}.

\begin{definition}\label{def:nele}
A \emph{NEoL-judgement} has the form
\begin{equation*}
  \fe\ent t\eq t':\sort
\end{equation*}
where $\fe\in\FE_\Sig$, $\sort\in\Sort_{\Sig}$ and $t,t'\in
\Term{\Sig}{\sort}{\seof{\fe}}$. Note that NEoL-judgements are also
NEL-judgements (Def. \ref{def:NELjud}).

A \emph{NEoL-theory} $\Th$ is a collection of such judgements, called its
axioms.
\end{definition}

\begin{definition}\label{def:neol_logc}
The set of \emph{theorems} of a NEoL-theory $\Th$ is the least set of judgements
containing the axioms of $\Th$ and closed under the rules of Fig.
\ref{Fig:rulneol}. We write
\begin{equation*}
  \fe\ento_{\Th} t \eq t':\sort
\end{equation*}
to indicate that the judgement is a theorem of $\Th$.

Say $\fe$ is as \eqref{eq:nel13}. Then the rule \ruleref{subst$^o$} in Fig.
\ref{Fig:rulneol} uses the following new pieces of notation (ref.
\eqref{eq:pro1} and Lem. \ref{lem:frsheqeq}):
\begin{itemize}
\item
  $\fe'\ent\sub\eq\sub'$ stands for the hypotheses $\fe'\ent\sub(x_i)\eq\sub'
  (x_i):\sort_i$ for $1\leq i\leq n$;
\item
  $\fe'\ent\sub:\fe$ stands for the hypotheses
  \begin{equation}
    \label{eq:pro51}
    (\fe')^{\sff\supp(\vec{a}'_i)}\ent\sub(x)\eq
      \swap{\vec{a}_i}{\vec{a}'_i}\acts\sub(x):\sort\enspace.
  \end{equation}
  for $1\leq i\leq n$, where $\vec{a}_i\in\Atom^{(n)}$ is an ordering of $\as_i$
  and $\vec{a}'_i\in\Atom^{(n)}$ is a tuple of the same size such that $\supp
  (\vec{a}'_i)\nsupp(\fe',\as_i,\sub(x))$.
\end{itemize}
\end{definition}

\begin{figure}
  \begin{mathpar}
    \inferrule*[left=(refl),right={$\fe\in\FE_{\Sig},
                                    t\in\Term{\Sig}{\sort}{\seof{\fe}}$}]%
    {%
      \ }{%
      \fe\ent t\eq t:\sort }
    \and
    \inferrule*[left=(symm$^o$)]%
    {%
      \fe\ent t\eq t':\sort }{%
      \fe\ent t'\eq t:\sort } \and
    \inferrule*[left=(trans$^o$)]%
    {%
      \fe\ent t\eq t':\sort\\
      \fe\ent t'\eq t'':\sort }{%
      \fe\ent t\eq t'':\sort } \and
    \inferrule*[left=(weak$^o$),right={$\fe\leq\fe'\in\FE_{\Sig}$}]%
    {
      \fe\ent t\eq t':\sort}{%
      \fe'\ent t\eq t':\sort} \and
    \inferrule*[left=(subst$^o$),
    right={$\sub,\sub':\seof{\fe}\to\seof{(\fe')}$}]%
    {%
      \fe'\ent \sub\eq\sub'\\
      \fe'\ent \sub:\fe \\
      \fe\ent t\eq t':\sort}{%
      \fe'\ent t\{\sub\}\eq t'\{\sub'\}:\sort} \and
    \inferrule*[left=(atm-elim$^o$),right={$\as\nsupp(\fe,t,t')$}]%
    {%
      \fe^{\sff\as}\ent t\eq t':\sort}{%
      \fe\ent t\eq t':\sort } \and
    \inferrule*[left={(susp)}]%
    {%
      \ }{%
      (\ds(\pi,\pi')\freshfor x:\sort)\ent\pi\,x\eq \pi'\,x:\sort }
  \end{mathpar}
  \caption{Proof rules for NEoL}
  \label{Fig:rulneol}
\end{figure}

\begin{theorem}\label{thm:neol_to_nel}
If $\Th$ is a NEoL-theory (and hence a NEL-theory) then $\fe\ento_{\Th}t\eq t':
\sort$ implies $\fe\ent_{\Th}t\eq t':\sort$.
\end{theorem}
\begin{proof}
We need only check that each of the rules for NEoL of Fig. \ref{Fig:rulneol} can
be derived from the rules for NEL of Fig. \ref{Fig:rulnel}. \ruleref{refl} and
\ruleref{susp} are also rules of Fig. \ref{Fig:rulnel}, while
\ruleref{symm$^o$}, \ruleref{trans$^o$}, \ruleref{weak$^o$} and
\ruleref{atm-elim$^o$} are clearly special cases of the corresponding rules.
\ruleref{subst$^o$} is a special case of \ruleref{subst}, as \eqref{eq:pro51} is
equivalent to the usual condition $\fe'\ent\as_i\nsupp\sub(x_i):\sort$ by Lem.
\ref{lem:frsheqeq}.
\end{proof}

The next three lemmas relate logical consequence for NEoL (Def.
\ref{def:neol_logc}) with the $\Perm$-actions on terms \eqref{eq:nel1} and
\eqref{eq:nel1x}.

\begin{lemma}\label{Lem:Th^oacts}
Given a NEoL-theory $\Th$, $\fe\ento_{\Th}t\eq t':\sort$ implies $\fe\ento_{\Th}
\pi\acts t\eq\pi\acts t':\sort$.
\end{lemma}
\begin{proof}
\begin{equation*}
\inferrule*[left=(subst$^o$)]
  { \fe\ent t\eq t':\sort \\
    \fe\ent \{t/x\}:(x:\sort) \\
    (x:\sort)\ent \pi\,x\eq\pi\,x:\sort }
  { \fe\ent \pi\,x\{t/x\}\eq\pi\,x\{t'/x\}:\sort }
\end{equation*}
\end{proof}

\begin{lemma}\label{Lem:Th^oact}
Given a NEoL-theory $\Th$, $\fe\ento_{\Th}t\eq t':\sort$ implies $\pi\act\fe
\ento_{\Th}\pi\act t\eq\pi\act t':\sort$.
\end{lemma}
\begin{proof}
By Lem. \ref{Lem:act_s}(ii) this result may be attained via \ruleref{susp$^o$}:
\begin{equation*}
  \inferrule
    {\pi\act\fe\ent(\pi^{-1}-):\fe \\
     \fe\ent\pi\acts t\eq \pi\acts t':\sort}
    {\pi\act\fe\ent(\pi\acts t)\{\pi^{-1}-\}\eq
      (\pi\acts t')\{\pi^{-1}-\}:\sort}
\end{equation*}
The second premise follows by Lem. \ref{Lem:Th^oacts}. Now take $\fe$ as
\eqref{eq:nel13} and for $1\leq i\leq n$ let $\vec{a}_i$ be an ordering of
$\as_i$ and $\vec{a}'_i$ be a suitably fresh tuple of the same size. Then $(\pi
\act\as_i\cup\supp(\vec{a}'_i)\freshfor x_i:\sort_i)\ent\pi^{-1}\,x_i\eq
\swap{\vec{a}_i}{\vec{a}'_i}\pi^{-1}\,x_i:\sort_i$ for each $i$ by
\ruleref{susp}; applying \ruleref{weak$^o$} gives us $(\pi\act\fe)^{\sff\supp
(\vec{a}'_i)}\ent\pi^{-1}\,x_i\eq\swap{\vec{a}_i}{\vec{a}'_i}\pi^{-1}\,x_i:
\sort_i$, which yields the first premise.
\end{proof}

\begin{lemma}\label{susp->perm}
Given $\fe\in\FE_{\Sig}$, $t\in\Term{\Sig}{\sort}{\seof{\fe}}$, and $\ds(\pi,
\pi')\nsupp t$,
\begin{equation*}
\fe^{\sff\ds(\pi,\pi')}\ento_{\Th}\pi\acts t\eq \pi'\acts t:\sort
\end{equation*}
\end{lemma}
\begin{proof}
By induction on the structure of $t$.

Suspensions: Say $t=\pi''\,x$. By our freshness assumption $\ds(\pi,\pi')=
\ds(\pi\pi'',\pi'\pi'')$, so by \ruleref{susp} $\ds(\pi,\pi')\freshfor x:\sort
\ent\pi\pi''\,x\eq\pi'\pi''\,x:\sort$; the result then follows by
\ruleref{weak$^o$}.

Constructed terms: if $t=\op\,t_1\cdots$ then $\pi\acts t=((\pi\act\op)\,x_1
\cdots)\{\sub\}$, where $\sub$ maps each $x_i\mapsto\pi\acts t_i$. Similarly
$\pi'\acts t=((\pi'\act\op)\,x_1\cdots)\{\sub'\}$, where $\sub'$ maps each $x_i\
\mapsto\pi'\acts t_i$. $\pi\act\op=\pi'\act\op$ by our freshness assumption and
Lem. \ref{lem:somefacts}(ii). To apply \ruleref{subst$^o$} to get our result we
need only then show $\fe^{\sff\ds(\pi,\pi')}\ent\sub\eq\sub'$; or for each $i$,
$\fe^{\sff\ds(\pi,\pi')}\ent\pi\acts t_i\eq\pi'\acts t_i:\sort_i$. These judgements
follow by induction.
\end{proof}

It is a fact about NEL that from $\fe\ent_{\Th}\as\freshfor t:\sort$ we can
infer $\fe\ent_{\Th}\as'\freshfor t:\sort$ for $\as'\subseteq\as$. The next
Lemma gives the corresponding result for NEoL.

\begin{lemma}\label{Lem:throw_fresh}
Suppose we have an NEoL-theory $\Th$, freshness environment $\fe\in\FE_{\Sig}$,
term $t\in\Term{\Sig}{\sort}{\seof{\fe}}$ and lists of atoms $\vec{a},\vec{b}
\in\Atom^{(n)}$ such that $\supp(\vec{b})\nsupp(\vec{a},t)$. Now suppose that
$\vec{a}',\vec{b}'\in\Atom^{(m)}$ for some $m\leq n$, with $\supp(\vec{a}')
\subseteq\supp(\vec{a})$ and $\supp(\vec{b}')\subseteq\supp(\vec{b})$. Then
\begin{equation*}
  \fe^{\sff\supp(\Vec{b})}\ento_{\Th}t\eq\swap{\vec{a}}{\vec{b}\,}\acts t:\sort
    \;\imp\; \fe^{\sff\supp(\Vec{b}')}\ento_{\Th}t\eq
    \swap{\vec{a}'}{\vec{b}'}\acts t:\sort
\end{equation*}
\end{lemma}
\begin{proof}
$(\supp(\Vec{a})\cup\supp(\Vec{b}')\freshfor x:\sort)\ent x\eq
\swap{\Vec{a}'}{\Vec{b}'}\,x:\sort$ by \ruleref{susp}. We wish to use
\ruleref{subst$^o$} to conclude that $\fe^{\sff\supp(\Vec{b}')}\ento_{\Th}x\{t/
x\}\eq(\swap{\vec{a}'}{\vec{b}'}\,x)\{t/x\}:\sort$; for this substitution to
occur we must prove that
\begin{equation}
  \label{eq:pro52a}
  \fe^{\sff\supp(\Vec{b}')}\ent\{t/x\}:(\supp(\Vec{a})\cup\supp(\Vec{b}')
    \freshfor x:\sort)\enspace.
\end{equation}
Now take fresh $\vec{c}\in\Atom^{(n)}$, $\vec{c}\,'\in\Atom^{(m)}$. By Lem.
\ref{susp->perm} we have 
\begin{equation}
  \label{eq:pro52b}
  \fe^{\sff\supp(\Vec{b}')\cup\supp(\vec{c}\,)\cup\supp(\vec{c}\,')}\ent
    \swap{\vec{a}}{\vec{c}}\acts t\eq
    \swap{\vec{a}}{\vec{c}}\swap{\vec{b}'}{\vec{c}\,'}\acts t:\sort\enspace.
\end{equation}
Applying Lem. \ref{Lem:Th^oact} and \ruleref{weak$^o$} to our hypothesis gives
us
\begin{equation}
  \label{eq:pro52c}
  \fe^{\sff\supp(\Vec{b}')\cup\supp(\vec{c}\,)\cup\supp(\vec{c}\,')}\ent t\eq
    \swap{\vec{a}}{\vec{c}}\acts t:\sort\enspace.
\end{equation}
Combining \eqref{eq:pro52b} and \eqref{eq:pro52c} with \ruleref{trans$^o$} gives
us
\begin{equation*}
  \fe^{\sff\supp(\Vec{b}')\cup\supp(\vec{c}\,)\cup\supp(\vec{c}\,')}\ent t\eq
    \swap{\vec{a}}{\vec{c}}\swap{\vec{b}'}{\vec{c}\,'}\acts t:\sort
\end{equation*}
which is equivalent to \eqref{eq:pro52a} as required.
\end{proof}

\begin{definition}\label{Def:Th^o}
Given a NEL-theory $\Th$, let $\Th^o$ be the NEoL-theory produced by replacing
each axiom of the form \eqref{eq:nel18} by the axioms
\begin{equation}
  \label{eq:pro54}
  \fe\ent t \eq t':\sort\text{ and }
    \fe^{\sff\supp(\Vec{a}')}\ent t\eq \swap{\Vec{a}}{\Vec{a}'}\acts t:\sort
\end{equation}
where $\vec{a}\in\Atom^{(n)}$ is an ordering of $\as$ and $\vec{a}'\in
\Atom^{(n)}$ is a tuple of the same size such that $\supp(\vec{a}')\nsupp(\fe,
\as,t)$. \eqref{eq:nel18} and \eqref{eq:pro54} are equivalent for NEL by Lem.
\ref{lem:frsheqeq}.
\end{definition}

\begin{theorem}\label{thm:big}
Let $\Th$ be a NEL-theory. Then $\fe\ent_{\Th}\as\freshfor t\eq t':\sort$
implies that $\fe\ento_{\Th^o}t\eq t':\sort$ and $\fe^{\sff\supp(\Vec{a}')}
\ento_{\Th^o}t\eq\swap{\Vec{a}}{\Vec{a}'}\acts t:\sort$, where $\Vec{a}\in
\Atom^{(n)}$ is an ordering of $\as$ and $\Vec{a}'\in\Atom^{(n)}$ is a tuple of
the same size such that $\supp(\Vec{a}')\nsupp(\fe,\as,t)$.

Therefore by Lem. \ref{lem:frsheqeq} and Thm. \ref{thm:neol_to_nel} NEL and NEoL
are equivalent.
\end{theorem}
\begin{proof}
Given a $\Th$-axiom the corresponding NEoL-judgement is a $\Th^o$-axiom by Def.
\ref{Def:Th^o}. The proof then proceeds by induction on the rules of Fig.
\ref{Fig:rulnel}, showing that the corresponding NEoL-judgements may be derived
by the rules of Fig. \ref{Fig:rulneol}.

This result is immediate for \ruleref{refl} and \ruleref{susp}, which are also
rules for NEoL. \ruleref{weak} and \ruleref{atm-elim} follow easily by
applications of \ruleref{weak$^o$} and \ruleref{atm-elim$^o$}.

\ruleref{symm}: Applying the induction hypothesis gives us $\fe\ent t\eq t':
\sort$ and $\fe^{\sff\supp(\vec{a}')}\ent t\eq\swap{\Vec{a}}{\Vec{a}'}\acts t:
\sort$. $\fe\ent t'\eq t:\sort$ by \ruleref{symm$^o$}. Now $\fe\ent
\swap{\vec{a}}{\vec{a}'}\acts t\eq\swap{\vec{a}}{\vec{a}'}\acts t':\sort$ by
Lem. \ref{Lem:Th^oacts}, so by \ruleref{weak$^o$} and \ruleref{trans$^o$} we
have $\fe^{\sff\supp(\vec{a}')}\ent t'\eq\swap{\Vec{a}}{\Vec{a}'}\acts t':\sort$.

\ruleref{trans}: The induction hypothesis gives us $\fe\ent t\eq t':\sort$,
$\fe^{\sff\supp(\vec{a}'_1)}\ent t\eq\swap{\vec{a}_1}{\vec{a}'_1}\acts t:\sort$,
$\fe\ent t'\eq t'':\sort$ and $\fe^{\sff\supp(\vec{a}'_2)}\ent t'\eq
\swap{\vec{a}_2}{\vec{a}'_2}\acts t':\sort$, where $\vec{a}_1,\vec{a}_2$ are
orderings of $\as_1,\as_2$ respectively, and $\vec{a}'_1,\vec{a}'_2$ are fresh
tuples of the same sizes. $\fe\ent t\eq t'':\sort$ by \ruleref{trans$^o$}. Now
suppose $\as_1\nsupp\as_2$ (if they are not disjoint we use Lem.
\ref{Lem:throw_fresh} to weaken one side until they are), and use successive
applications of \ruleref{weak$^o$}, Lem. \ref{Lem:Th^oacts} and
\ruleref{trans$^o$}:
\begin{equation*}
\begin{array}{rcl}
  \fe^{\sff\supp(\Vec{a}'_1)\cup\supp(\Vec{a}'_2)}\ento_{\Th}\,t &\eq&
    \swap{\vec{a}_1}{\vec{a}'_1}\acts t \\
  &\eq& \swap{\vec{a}_1}{\vec{a}'_1}\acts t' \\
  &\eq& \swap{\vec{a}_1}{\vec{a}'_1}\swap{\vec{a}_2}{\vec{a}'_2}\acts t' \\
  &\eq&
    \swap{\vec{a}_1}{\vec{a}'_1}\swap{\vec{a}_2}{\vec{a}'_2}\acts t\enspace.
\end{array}
\end{equation*}

\ruleref{subst}: By the induction hypothesis $\fe\ent t\eq t':\sort$ and
$\fe^{\sff\supp(\Vec{a}')}\ent t\eq\swap{\Vec{a}}{\Vec{a}'}\acts t:\sort$, and if
$\fe$ is as \eqref{eq:nel13} then for $1\leq i\leq n$ we have $\fe'\ent\sub(x_i)
\eq\sub'(x_i):\sort_i$ and $(\fe')^{\sff\supp(\vec{a}'_i)}\ent\sub(x_i)\eq
\swap{\vec{a}_i}{\vec{a}'_i}\acts\sub(x_i):\sort_i$ where $\vec{a}_i$ is an
ordering of $\as_i$ and $\vec{a}'_i$ is a fresh tuple of the same size. $\fe\ent
t\{\sub\}\eq t'\{\sub'\}:\sort$ by \ruleref{subst$^o$}. Now
$\swap{\vec{a}}{\vec{a}'}\acts(t\{\sub\})=(\swap{\vec{a}}{\vec{a}'}\acts t)
\{\sub\}$ by Lem. \ref{Lem:act_s}(i), so we look to apply \ruleref{subst$^o$}:
\begin{equation*}
  \inferrule
  { (\fe')^{\sff\supp(\vec{a}')}\ent\sub:\fe^{\sff\supp(\vec{a}')} \\
    \fe^{\sff\supp(\vec{a}')}\ent t\eq\swap{\Vec{a}}{\Vec{a}'}\acts t:\sort }
  { (\fe')^{\sff\supp(\vec{a}')}\ent t\{\sub\}\eq
      (\swap{\vec{a}}{\vec{a}'}\acts t)\{\sub\}:\sort }
\end{equation*}
The second premise is among our hypotheses, while the first follows from
\ruleref{subst$^o$} for each $i$:
\begin{equation*}
  \inferrule
  { \fe''\ent\sub(x_i)\eq\swap{\vec{a}_i}{\vec{a}'_i}\acts\sub(x_i):\sort_i \\
    \fe''\ent\{\sub(x_i)/x\}:
      (\supp(\vec{a}')\cup\supp(\vec{a}'')\freshfor x:\sort_i) \\
    (\supp(\vec{a}')\cup\supp(\vec{a}'')\freshfor x:\sort)\ent x\eq
      \swap{\vec{a}'}{\vec{a}''}\,x:\sort_i }
  { \fe''\ent x\{\sub(x_i)/x\}\eq\swap{\vec{a}'}{\vec{a}''}\,x
      \{\swap{\vec{a}_i}{\vec{a}'_i}\acts\sub(x_i)/x\}:\sort_i }
\end{equation*}
where $\fe''=(\fe')^{\sff\supp(\vec{a}')\,\cup\,\supp(\vec{a}'')\,\cup\,\supp
(\vec{a}_i')}$ and $\vec{a}''$ is a fresh copy of $\vec{a}'$. The first premise
here follows from our hypotheses and \ruleref{weak$^o$}; the second follows by
Lem. \ref{susp->perm}, and the third by \ruleref{susp}.

\ruleref{atm-intro}: By the induction hypothesis $\fe\ent t\eq t':\sort$ and
$\fe^{\sff\supp(\vec{b}')}\ent t\eq\swap{\vec{a}'}{\vec{b}'}\acts t:\sort$, where
$\vec{a}'$ is an ordering of $\as'$ and $\vec{b}'$ is a fresh tuple of the same
size. $\fe^{\sff\as}\ent t\eq t'$ by \ruleref{weak$^o$}. We need to prove that
$\fe'\ent t\eq\swap{\vec{a}}{\vec{b}}\swap{\vec{a}'}{\vec{b}'}\acts t:\sort$,
where $\fe'=\fe^{\sff\as\,\cup\,\supp(\vec{b})\,\cup\,\supp(\vec{b}')}$, $\vec{a}$ is
an ordering of $\as$ and $\vec{b}$ is a fresh copy. Apply \ruleref{subst$^o$}:
\begin{equation*}
  \inferrule
  { \fe'\ent t\eq\swap{\vec{a}}{\vec{b}'}\acts t:\sort \\
    \fe'\ent\{t/x\}:(\as\cup\supp(\vec{b})\freshfor x:\sort) \\
    (\as\cup\supp(\vec{b})\freshfor x:\sort)\ent x\eq
      \swap{\vec{a}}{\vec{b}}\,x:\sort }
  { \fe'\ent x\{t/x\}\eq
      (\swap{\vec{a}}{\vec{b}}\,x)\{\swap{\vec{a}}{\vec{b}'}\acts t/x\}:\sort }
\end{equation*}
The first premise follows from our hypothesis and \ruleref{weak$^o$}; the second
follows by Lem. \ref{susp->perm}, and the third by \ruleref{susp}.

\ruleref{$\freshfor$-equivar}: $(\as\cup\supp(\vec{a}')\freshfor x:\sort)\ent\pi
\,x\eq\swap{\pi\act\vec{a}}{\vec{a}'}\pi\,x:\sort$ by \ruleref{susp}.
\end{proof}

%%%%%%%%%%%%%%%%%%%%%%%%%%%%%%%%%%%%%%%%%%%%%%%%%%%%%%%%%%%%%%%%%%%%%%%%%%%%%%%
\section{The Empty Theory}\label{sec:empty}

In standard equational logic, provable equality corresponds to literal syntactic
equality in the empty theory $\emptyset$ without axioms. This is not the case
for NEoL because of the presence of freshness environments and suspensions. In
this section we give a simple syntax-directed description of equality in the
empty NEoL-theory. It is then straightforward to extend this to a description of
freshness in the empty NEL-theory.

\begin{figure}
  \begin{mathpar}
    \inferrule*
    { \ds(\pi,\pi')\freshfor x:\sort\in\fe }
    { \fe\ent \pi\,x\eq\pi'\,x:\sort } \and
    \inferrule*
    { \fe\ent t_1\eq t'_1:\sort_1 \\
      \cdots \\
      \fe\ent t_n\eq t'_n:\sort_n }
    { \fe\ent \op\,t_1\cdots t_n\eq\op\,t'_1\cdots t'_n:\sort }
  \end{mathpar}
  \caption{Syntax-directed rules for Nominal Equality}
  \label{Fig:empty}
\end{figure}

\begin{definition}
Fig. \ref{Fig:empty} provides syntax-directed rules for Nominal Equality. The
notation `$\ds(\pi,\pi')\freshfor x:\sort\in\fe$' means that $\fe(x)=(\as,
\sort)$ for $\as\supseteq\ds(\pi,\pi')$.
\end{definition}

\begin{theorem}\label{thm:neol_to_sdr}
The rules for NEoL of Fig. \ref{Fig:rulneol} imply the syntax-directed rules of
Fig. \ref{Fig:empty}.
\end{theorem}
\begin{proof}
The suspension case hold by \ruleref{susp} and \ruleref{weak$^o$}; the
constructed term case by \ruleref{subst$^o$}.
\end{proof}

\begin{lemma}\label{lem:tech}
Say we have permutations $\pi,\pi'\in\Perm$ and a finite list $\vec{a}$ of atoms
so that $\ds(\pi,\pi')\subseteq\supp(\vec{a})$. Let $\vec{a}'$ be a list of
fresh atoms of the same size. Then if we can derive
\begin{equation*}
  \fe\ent t\eq t':\sort \mbox{ and }
  \fe^{\sff\supp(\vec{a}')}\ent t\eq\swap{\vec{a}}{\vec{a}'}\acts t:\sort
\end{equation*}
by the syntax-directed rules of Fig. \ref{Fig:empty}, then we can also derive
\begin{equation*}
  \fe\ent \pi\acts t\eq \pi'\acts t':\sort
\end{equation*}
\end{lemma}
\begin{proof}
Suspensions: Say $t=\xi\,x$ and $t'=\xi'\,x$, so $\ds(\xi,\xi')\freshfor x:\sort
\in\fe$ and $\ds(\xi,\swap{\vec{a}}{\vec{a}'}\xi)\freshfor x:\sort\in\fe^{\sff
\supp(\vec{a}')}$. We must prove that $\fe\ent\pi\xi\,x\eq\pi'\xi'\,x$, i.e.
that $\ds(\pi\xi,\pi'\xi')\freshfor x:\sort\in\fe$. Take $a\in\ds(\pi\xi,\pi'
\xi')$. If $a\in\ds(\xi,\xi')$ we're done by our first assumption. But if $\xi
(a)\in\ds(\pi,\pi')$ then $\xi(a)\in\supp(\vec{a})$, so $a\in\ds(\xi,
\swap{\vec{a}}{\vec{a}'}\xi)$ and by our second assumption $a\freshfor x:\sort
\in\fe^{\sff\supp(\vec{a}')}$. But $\vec{a}'$ was chosen fresh, so $a\freshfor x\in
\fe$.

Constructed terms: Let $t=\op\,t_1\cdots$ and $t'=\op\,t'_1\cdots$, so for all
$i$, $\fe\ent t_i\eq t'_i:\sort_i$, $\fe^{\sff\supp(\vec{a}')}\ent t_i\eq
\swap{\vec{a}}{\vec{a}'}\acts t_i:\sort_i$ and $\op=\swap{\vec{a}}{\vec{a}'}\act
\op$. We must prove that $\fe\ent(\pi\act\op)(\pi\acts t_1)\cdots\eq(\pi'\act
\op)(\pi'\acts t'_1)\cdots$. Now $\supp(\vec{a}')\nsupp\op$, so by Lem.
\ref{lem:somefacts}(i) $\supp(\vec{a})\nsupp\swap{\vec{a}}{\vec{a}'}\act\op=
\op$. $\ds(\pi,\pi')\subseteq\supp(\vec{a})$, so by Lem. \ref{lem:somefacts}(ii)
$\pi\act\op=\pi'\act\op$. Finally, $\fe\ent\pi\acts t_i\eq\pi'\acts t'_i:
\sort_i$ follows by induction.
\end{proof}

\begin{theorem}\label{thm:emptyexplained}
Suppose $\fe\ento_{\emptyset}t\eq t':\sort$ by the rules for NEoL of Fig.
\ref{Fig:rulneol}. Then $\fe\ent t\eq t':\sort$ by the syntax-directed rules of
Fig. \ref{Fig:empty}.

Therefore by Thm. \ref{thm:neol_to_sdr} the syntax-directed rules coincide with
the empty NEoL-theory.
\end{theorem}
\begin{proof}
By induction on the rules of Fig. \ref{Fig:rulneol}.

\ruleref{susp} follows immediately from the suspension case of Fig.
\ref{Fig:empty}. \ruleref{refl} follows by an easy induction on the structure of
$t$. \ruleref{symm$^o$}, \ruleref{weak$^o$} and \ruleref{atm-elim$^o$} are also
straightforward.

\ruleref{trans$^o$}: Say $t=\pi\,x$, $t'=\pi'\,x$ and $t''=\pi''\,x$, so $\ds
(\pi,\pi')\freshfor x:\sort,\,\ds(\pi',\pi'')\freshfor x:\sort\in\fe$. We need
to show that $\ds(\pi,\pi'')\freshfor x:\sort\in\fe$, so take $a\in\ds(\pi,
\pi'')$. If $a\in\ds(\pi,\pi')$ we are done by our first assumption, but if $\pi
(a)=\pi'(a)$ then $a\in\ds(\pi',\pi'')$ so we are done by our second assumption.
The constructed term case is an easy induction.

\ruleref{subst$^o$}: Say $\fe\ent t\eq t'$ is $\fe\ent\pi\,x\eq\pi\,x:\sort$,
$\sub(x)=\xi\,y$ and $\sub'(x)=\xi'\,y$, so $\ds(\pi,\pi')\freshfor x:\sort\in
\fe$ and $\ds(\xi,\xi')\freshfor y:\sort\in\fe'$. The other premise
\eqref{eq:pro51} says that $\ds(\xi,\swap{\vec{a}}{\vec{a}'}\xi)\freshfor y:
\sort\in(\fe')^{\sff\supp(\vec{a}')}$ where $\fe(x)=(\as,\sort)$, $\vec{a}$ is an
ordering of $\as$, and $\vec{a}'$ is a fresh tuple of the same size. We need to
prove that $\ds(\pi\xi,\pi'\xi')\freshfor y:\sort\in\fe'$. Take $a\in\ds(\pi\xi,
\pi'\xi')$. If $a\in\ds(\xi,\xi')$ we are done, so say $\xi(a)\in\ds(\pi,\pi')
\subseteq\as$. Then $a\in\ds(\xi,\swap{\vec{a}}{\vec{a}'}\xi)$, so $a\freshfor
y:\sort\in(\fe')^{\sff\supp(\vec{a}')}$, but $\vec{a}'$ was chosen fresh, so we are
done.

Now take $t,t'$ as above, so $\ds(\pi,\pi')\freshfor x:\sort\in\fe$ still, but
$\sub(x)=\op\,t_1\cdots$ and $\sub'(x)=\op\,t'_1\cdots$. Then for $1\leq i\leq
n$, $\fe'\ent t_i\eq t'_i:\sort_i$, $(\fe')^{\sff\supp(\vec{a}')}\ent t_i\eq
\swap{\vec{a}}{\vec{a}'}\acts t_i:\sort$ and $\op=\swap{\vec{a}}{\vec{a}'}\act
\op$, where $\as,\vec{a},\vec{a}'$ are also as above. $\supp(\vec{a}')\nsupp\op$
implies $\as\nsupp\swap{\vec{a}}{\vec{a}'}\act\op=\op$, so $\pi\act\op\eq\pi'
\act\op$. $\fe'\ent\pi\acts t_i\eq\pi\acts t'_i$ by Lem. \ref{lem:tech}.

Finally, if $t,t'$ are constructed terms then the induction for
\ruleref{subst$^o$} is straightforward.
\end{proof}

\begin{corollary}\label{cor:freshback}
The empty NEL-theory, following the rules of Fig. \ref{Fig:rulnel}, coincides
with the syntax-directed rules of Fig. \ref{Fig:empty} along with these new
rules for freshness:
\begin{mathpar}
  \inferrule*
  { \pi^{-1}\act\as\freshfor x:\sort\in\fe }
  { \fe\ent \as\freshfor\pi\,x:\sort } \and
  \inferrule*
  { \fe\ent \as\freshfor t_1:\sort_1\\
    \cdots\\
    \fe\ent \as\freshfor t_n:\sort_n\\
    \as\nsupp\op }
  { \fe\ent \as\freshfor \op\,t_1\cdots t_n:\sort }
\end{mathpar}
\end{corollary}
\begin{proof}
Lem. \ref{lem:frsheqeq}, Thm. \ref{thm:big} and Thm. \ref{thm:emptyexplained}.
\end{proof}

%%%%%%%%%%%%%%%%%%%%%%%%%%%%%%%%%%%%%%%%%%%%%%%%%%%%%%%%%%%%%%%%%%%%%%%%%%%%%%%
\section{Related Work}\label{sec:related}

\textbf{Equational logic for nominal unification.} The first notion of
equational logic over nominal sets to be developed were the syntax-directed
rules of \cite[Fig. 2]{Urban:Nominal}, which were used in the definition of
nominal unification. The syntax that directs this definition is based on
\emph{nominal signatures}, which compared to the signatures of Def.
\ref{def:sig} have a richer sort system and a set, rather than nominal set, of
operation symbols.

The rules ($\eq$-suspension), ($\eq$-function symbol), ($\freshfor$-suspension)
and ($\freshfor$-function symbol) of \cite{Urban:Nominal} clearly match the
syntax-directed rules of Fig. \ref{Fig:empty} and Cor. \ref{cor:freshback},
apart from the premise $\as\nsupp\op$, which is non-trivial only when $\op$ may
have non-empty support. If we add operation symbols for unit, pairing, atoms and
atom-abstraction then, via Fig. \ref{Fig:empty} and Cor.
\ref{cor:freshback}, we recover all of the rules of \cite{Urban:Nominal} except
for ($\eq$-abstraction-2) and ($\freshfor$-abstraction-1): 
\begin{mathpar}
  \inferrule*
  { a \neq a'\\
    \fe\ent t\eq \swap{a}{a'}\acts t' \\
    \fe\ent a\freshfor t' }
  { \fe\ent a.\,t \eq a'.\,t' } \and
  \inferrule*
  { \ }
  { \fe\ent a\freshfor a.\,t }
\end{mathpar}
where $a.\,t$ is the atom-abstraction binding $a$ in $t$. These are the rules
for $\alpha$-equivalence. Following \cite{Clouston:Binding}, we may capture
these rules by moving from the empty theory to the theory with one axiom
$$
(b\freshfor x)\ent a.\,x\eq b.\,\swap{a}{b}\,x
$$
or equivalently, $(x)\ent a\freshfor a.\,x$.

\textbf{Nominal Algebra (NA).} NA \cite{Gabbay:Nominal} is a logic independently
developed to reason about the same properties as NEL, but with some
interestingly different design choices. NA is built on nominal signatures, so in
the empty theory equality should be, as above, $\alpha$-equivalence over these
signatures rather than the weaker equivalence of Sec. \ref{sec:empty}. NA
employs a set, rather than nominal set, of operation symbols, which may make it
less expressive than NEL. For example, with NEL one could define a nominal set
of operation symbols isomorphic to the nominal set $\Atom^{(2)}$ of disjoint
pairs of atoms; this does not seem to be possible with NA.

Finally, NA employs a syntax-directed notion of freshness that is weaker than
that used by NEL; in particular the transitive property
\begin{mathpar}
  \inferrule*
  { \fe\ent a\freshfor t:\sort \\
    \fe\ent t\eq t':\sort }
  { \fe\ent a\freshfor t':\sort }
\end{mathpar}
does not hold. In \cite[Sec. 5]{Gabbay:Nominal} design alternatives for NA were
discussed where atom-abstraction sorts were eliminated and the freshness
relation strengthened to match that of NEL. However no notion of NA with
equality only has been proposed along the lines of NEoL, although it seems
likely that such a logic could be defined by working in close analogy with the
results of this paper.

\textbf{Synthetic Nominal Equational Logic (SNEL).} Term Equational Systems
\cite{Fiore:Term} are a category theoretic account of equational logic,
including proof theory. This framework allows equational logic to be naturally
generalised from the category of sets to other categories, with proof rules
automatically generated in each new setting so long as the new categories obey
certain certain constraints. Following NEL and NA, Term Equational Systems were
developed in the category of nominal sets, and the resulting logic is called
SNEL \cite[Sec. 5]{Fiore:Term}. SNEL is another notion of nominal sets with
equations only, but no proof was offered that the addition of  freshness
judgements would not strengthen the logic. The authors were, however, aware of
the results presented in this paper, which could be seen as a sanity check on
the development of the equation-only SNEL.

It should also be noted that the syntax of SNEL is not entirely in keeping with
that which is commonly used in nominal logic, as we have no freshness
environments or suspensions. For example, the axiom \eqref{eq:nel5} for
$\alpha$-equivalence in the untyped $\lambda$-calculus would be written
$$
[a,b]\{x:1\}\ent \op[lam]_a\,x(a) \eq \op[lam]_b\,x(b)\enspace.
$$
Here the metavariables explicitly refer to names they may depend on. This
differs from the standard mathematical treatment of bound names, which most
applications of nominal techniques try to capture. It is an interesting question
whether a more standard presentation of equational logic over nominal sets, such
as NEL or NA, could be derived in this category theoretic context.

%%%%%%%%%%%%%%%%%%%%%%%%%%%%%%%%%%%%%%%%%%%%%%%%%%%%%%%%%%%%%%%%%%%%%%%%%%%%%%%
\nocite{*}
\bibliographystyle{eptcs}
\bibliography{LFMTP_bib}

\begin{thebibliography}{10}
\providecommand{\bibitemdeclare}[2]{}
\providecommand{\urlprefix}{Available at }
\providecommand{\url}[1]{\texttt{#1}}
\providecommand{\href}[2]{\texttt{#2}}
\providecommand{\urlalt}[2]{\href{#1}{#2}}
\providecommand{\doi}[1]{doi:\urlalt{http://dx.doi.org/#1}{#1}}
\providecommand{\bibinfo}[2]{#2}

\bibitemdeclare{book}{Burris:Logic}
\bibitem{Burris:Logic}
\bibinfo{author}{S.~N. Burris} (\bibinfo{year}{1998}):
  \emph{\bibinfo{title}{Logic for Mathematics and Computer Science}}.
\newblock \bibinfo{publisher}{Prentice Hall}.

\bibitemdeclare{phdthesis}{Clouston:Equational}
\bibitem{Clouston:Equational}
\bibinfo{author}{R.~Clouston} (\bibinfo{year}{2009}):
  \emph{\bibinfo{title}{Equational Logic for Names and Binders}}.
\newblock Ph.D. thesis, \bibinfo{school}{University of Cambridge}.
\newblock
  \bibinfo{note}{\url{http://cecs.anu.edu.au/~rclouston/Clouston_Thesis.pdf}.}

\bibitemdeclare{inproceedings}{Clouston:Binding}
\bibitem{Clouston:Binding}
\bibinfo{author}{R.~Clouston} (\bibinfo{year}{2010}):
  \emph{\bibinfo{title}{Binding in Nominal Equational Logic}}.
\newblock In: {\sl \bibinfo{booktitle}{MFPS}}, {\sl \bibinfo{series}{ENTCS}}
  \bibinfo{volume}{265}, pp. \bibinfo{pages}{259--276},
  \doi{10.1016/j.entcs.2010.08.016}.

\bibitemdeclare{inproceedings}{Clouston:Lawvere}
\bibitem{Clouston:Lawvere}
\bibinfo{author}{R.~Clouston} (\bibinfo{year}{2011}):
  \emph{\bibinfo{title}{Nominal {L}awvere Theories}}.
\newblock In: {\sl \bibinfo{booktitle}{WoLLIC}}, {\sl \bibinfo{series}{LNCS}}
  \bibinfo{volume}{6642}, pp. \bibinfo{pages}{67--83},
  \doi{10.1007/978-3-642-20920-8\_11}.

\bibitemdeclare{article}{Clouston:Nominal}
\bibitem{Clouston:Nominal}
\bibinfo{author}{R.~Clouston} \& \bibinfo{author}{A.~M. Pitts}
  (\bibinfo{year}{2007}): \emph{\bibinfo{title}{Nominal Equational Logic}}.
\newblock {\sl \bibinfo{journal}{ENTCS}} \bibinfo{volume}{172}, pp.
  \bibinfo{pages}{223--257}, \doi{10.1016/j.entcs.2010.08.016}.

\bibitemdeclare{inproceedings}{Fiore:Term}
\bibitem{Fiore:Term}
\bibinfo{author}{M.~Fiore} \& \bibinfo{author}{C-K. Hur}
  (\bibinfo{year}{2008}): \emph{\bibinfo{title}{Term Equational Systems and
  Logics}}.
\newblock In: {\sl \bibinfo{booktitle}{MFPS}}, {\sl \bibinfo{series}{ENTCS}}
  \bibinfo{volume}{218}, pp. \bibinfo{pages}{171--192},
  \doi{10.1016/j.entcs.2008.10.011}.

\bibitemdeclare{inproceedings}{Gabbay:One}
\bibitem{Gabbay:One}
\bibinfo{author}{M.~J. Gabbay} \& \bibinfo{author}{A.~Mathijssen}
  (\bibinfo{year}{2006}): \emph{\bibinfo{title}{One-and-a-halfth-order Logic}}.
\newblock In: {\sl \bibinfo{booktitle}{PPDP}}, \bibinfo{publisher}{ACM}, pp.
  \bibinfo{pages}{189--200}, \doi{10.1145/1140335.1140359}.

\bibitemdeclare{article}{Gabbay:Nominal}
\bibitem{Gabbay:Nominal}
\bibinfo{author}{M.~J. Gabbay} \& \bibinfo{author}{A.~Mathijssen}
  (\bibinfo{year}{2009}): \emph{\bibinfo{title}{Nominal (Universal) Algebra:
  Equational Logic with Names and Binding}}.
\newblock {\sl \bibinfo{journal}{J. Logic Comput.}}
  \bibinfo{volume}{19}(\bibinfo{number}{6}), pp. \bibinfo{pages}{1455--1508},
  \doi{10.1007/s001650200016}.

\bibitemdeclare{article}{Gabbay:New}
\bibitem{Gabbay:New}
\bibinfo{author}{M.~J. Gabbay} \& \bibinfo{author}{A.~M. Pitts}
  (\bibinfo{year}{2002}): \emph{\bibinfo{title}{A New Approach to Abstract
  Syntax with Variable Binding}}.
\newblock {\sl \bibinfo{journal}{FAC}} \bibinfo{volume}{13}, pp.
  \bibinfo{pages}{341--363}, \doi{10.1016/j.entcs.2008.10.011}.

\bibitemdeclare{phdthesis}{Lawvere:Functorial}
\bibitem{Lawvere:Functorial}
\bibinfo{author}{F.~W. Lawvere} (\bibinfo{year}{1963}):
  \emph{\bibinfo{title}{Functorial Semantics of Algebraic Theories}}.
\newblock Ph.D. thesis, \bibinfo{school}{Columbia University}.

\bibitemdeclare{article}{Pitts:Nominal}
\bibitem{Pitts:Nominal}
\bibinfo{author}{A.~M. Pitts} (\bibinfo{year}{2001}):
  \emph{\bibinfo{title}{Nominal Logic: a First Order Theory of Names and
  Binding}}.
\newblock {\sl \bibinfo{journal}{LNCS}} \bibinfo{volume}{2215}, pp.
  \bibinfo{pages}{219--242}, \doi{10.1007/3-540-45500-0_11}.

\bibitemdeclare{article}{Pitts:Alpha}
\bibitem{Pitts:Alpha}
\bibinfo{author}{A.~M. Pitts} (\bibinfo{year}{2006}):
  \emph{\bibinfo{title}{Alpha-structural Recursion and Induction}}.
\newblock {\sl \bibinfo{journal}{J. ACM}} \bibinfo{volume}{53}, pp.
  \bibinfo{pages}{459--506}, \doi{10.1145/1147954.1147961}.

\bibitemdeclare{article}{Urban:Nominal}
\bibitem{Urban:Nominal}
\bibinfo{author}{C.~Urban}, \bibinfo{author}{A.~M. Pitts} \&
  \bibinfo{author}{M.~J. Gabbay} (\bibinfo{year}{2004}):
  \emph{\bibinfo{title}{Nominal Unification}}.
\newblock {\sl \bibinfo{journal}{TCS}}
  \bibinfo{volume}{323}(\bibinfo{number}{1-3}), pp. \bibinfo{pages}{473--497},
  \doi{10.1007/s10817-009-9164-3}.

\end{thebibliography}
\end{document}